\documentclass{IEEEtran}

\usepackage{amssymb, amsbsy, amsmath, amsfonts, amsthm}
\usepackage{graphicx, cite, mathtools}
\usepackage{dsfont}
\usepackage{bbm}
\usepackage{textcomp}
\usepackage{xcolor}
\usepackage{setspace}

\newcommand\PR{\mathrm{PR}}
\newcommand\AR{\mathrm{AR}}
\newcommand\FR{\mathrm{FR}}
\newcommand\CLT{\mathrm{CLT}}

\theoremstyle{plain}
\newtheorem{lem}{Lemma}
\newtheorem{prop}{Proposition}
\newtheorem{thm}{Theorem}
\newtheorem{cor}{Corollary}

\normalsize
\allowdisplaybreaks
\begin{document}
	
	\title{Outage and DMT Analysis \\of Partition-based Schemes for RIS-aided \\MIMO Fading Channels}
	
	\author{Andreas Nicolaides, \textit{Student Member, IEEE,}
		Constantinos Psomas, \textit{Senior Member, IEEE,}
		Ghassan M. Kraidy, \textit{Senior Member, IEEE,}
		Sheng Yang, \textit{Member, IEEE,}
		and~Ioannis Krikidis, \textit{Fellow, IEEE}% <-this % stops a space
		\thanks{This work was co-funded by the European Regional Development Fund and the Republic of Cyprus through the Research and Innovation Foundation, under the projects ENTERPRISES/0521/0068 (AGRILORA) and INFRASTRUCTURES/1216/0017 (IRIDA). It has also received funding from the European Research Council (ERC) under the European Union’s Horizon 2020 research and innovation programme (Grant agreement No. 819819).}
		\thanks{A. Nicolaides, C. Psomas and I. Krikidis are with the IRIDA Research Centre for Communication Technologies, Department of Electrical and Computer Engineering, University of Cyprus, Cyprus (e-mail: \{anicol09, psomas, krikidis\}@ucy.ac.cy). G. M. Kraidy is with the Department of Electronic Systems, Norwegian University of Science and Technology, Norway (email: ghassan.kraidy@ntnu.no). This work was performed when he was with the IRIDA Research Centre for Communication Technologies. S. Yang is with the Laboratory of Signals and Systems, CentraleSup\'elec, Paris-Saclay University, France (email: sheng.yang@centralesupelec.fr). Part of this work was presented at the IEEE International Symposium on Information Theory, Espoo, Finland, June 2022 \cite{nicolaides2022}.}}

\maketitle

\pagestyle{empty}

\thispagestyle{empty}

\begin{abstract}
	In this paper, we investigate the performance of multiple-input multiple-output (MIMO) fading channels assisted by a reconfigurable intelligent surface (RIS), through the employment of partition-based RIS schemes. The proposed schemes are implemented without requiring any channel state information knowledge at the transmitter side; this characteristic makes them attractive for practical applications. In particular, the RIS elements are partitioned into sub-surfaces, which are periodically modified in an efficient way to assist the communication. Under this framework, we propose two low-complexity partition-based schemes, where each sub-surface is adjusted by following an amplitude-based or a phase-based approach. Specifically, the \textit{activate-reflect} (AR) scheme activates each sub-surface consecutively, by changing the reflection amplitude of the corresponding elements. On the other hand, the \textit{flip-reflect} (FR) scheme adjusts periodically the phase shift of the elements at each sub-surface. Through the sequential reconfiguration of each sub-surface, an equivalent parallel channel in the time domain is produced. We analyze the performance of each scheme in terms of outage probability and provide expressions for the achieved diversity-multiplexing tradeoff. Our results show that the asymptotic performance of the considered network under the partition-based schemes can be significantly enhanced in terms of diversity gain compared to the conventional case, where a single partition is considered. Moreover, the FR scheme always achieves the maximum multiplexing gain, while for the AR scheme this maximum gain can be achieved only under certain conditions with respect to the number of elements in each sub-surface.
\end{abstract}


\IEEEpeerreviewmaketitle

\section{Introduction}
	The development of future generations of wireless communications is envisioned to satisfy the constantly increasing demands in the number of devices that need to communicate, with extremely high data rates and ultra-reliable connectivity capabilities \cite{zhang2019}. In particular, recent research advances towards the realization of the 6G era suggest that, by successfully controlling the wireless propagation environment in an efficient and intelligent manner, the performance of wireless networks could be enhanced beyond the current limits \cite{renzo2019}. Towards this direction, reconfigurable intelligent surfaces (RISs) have been proposed as a potential and appealing solution \cite{pan2021}. An RIS is a planar metasurface equipped with a large number of passive reflecting elements that are connected to a smart controller, where each element can modify the phase and/or induce an amplitude attenuation to the incident signal \cite{wu2021}. Their employment is therefore a cost effective solution, which can improve energy efficiency, due to the passive operation of the elements, and spectral efficiency, since they operate in ideal full-duplex mode, as well as increase the coverage of wireless networks.
	
	Due to these performance benefits, the use of RISs has been considered in several applications, and the performance of RIS-aided wireless networks has been investigated for various communication scenarios. Specifically, the authors in \cite{tao2020} studied the performance of a single-input single-output (SISO) system, and provided tight closed-form approximations for fundamental metrics such as the ergodic capacity and the outage probability. In \cite{chengjun2021}, the system performance of RIS-aided orthogonal multiple access (OMA) and non-orthogonal multiple access (NOMA) SISO networks was studied in terms of outage probability and ergodic rate. It was shown that such networks obtain significant performance gains with the employment of an RIS, which is superior to the use of full-duplex decode-and-forward relays. Furthermore, the work in \cite{psomas2021} proposed RIS-enabled random-rotation schemes for SISO networks, which can be implemented with limited or without channel state information (CSI) knowledge. It was demonstrated that the presented schemes improve the performance in terms of energy efficiency, outage probability and diversity gain. The diversity order achieved by an RIS-aided SISO system was also studied in \cite{xu2021}, where the authors derived the minimum number of required quantization levels for the RIS discrete phase shifts to achieve full diversity. In \cite{atapattu2020}, the authors introduced an RIS-assisted system for two-way communications and proposed two transmission schemes, for which they investigated the performance limits in terms of outage probability and spectral efficiency.
	
	The employment of RISs has also been considered to assist the communication in wireless networks with multiple antennas at the transmitter and/or the receiver. In particular, a single-cell wireless network, where an RIS is deployed to assist the communication between a multi-antenna access point (AP) and multiple single-antenna users, was investigated in \cite{wu2019}. It was demonstrated that the system's performance can be enhanced in terms of both spectral and energy efficiency, by jointly optimizing the active and passive beamforming vectors at the AP and the RIS, respectively. The implementation of RISs has been also considered under index modulation schemes, enabling simultaneous passive beamforming and information transfer of the RIS in multiple-input single-output (MISO) systems \cite{lin2021,lin2022}. Specifically, in \cite{lin2021}, the authors introduced an amplitude-based scheme, called the reflection pattern modulation (RPM) scheme, where the joint passive beamforming and information transfer occured by activating different subsets of elements at the RIS. It was shown that the presented approach achieves great improvement of the achievable rate performance. In \cite{lin2022}, an RIS reflection modulation scheme was proposed, referred to as the quadrature reflection modulation (QRM) scheme, which was proven to outperform the RPM scheme. Moreover, a multi-hop RIS-enabled scenario was considered in \cite{huang2021}, where multiple RISs were deployed to assist the communication between a multi-antenna AP and multiple users with single antennas, in order to improve network coverage of terahertz communications.
	
	The fundamental capacity limit for an RIS-aided multiple-input multiple-output (MIMO) network was provided in \cite{zhang2020}, by developing algorithms for jointly optimizing the RIS reflection coefficients and the transmit covariance matrix. The same optimization parameters were also considered for multi-cell communications in \cite{pan2020}, where the weighted sum-rate of a multi-cell multi-user MIMO system was enhanced by employing an RIS at the cell boundary. Furthermore, the asymptotic performance of an RIS-aided MIMO channel was investigated in \cite{cheng2021}, where it was demonstrated that the achieved multiplexing gain can be enhanced if the information data stream is available both at the transmitter and the RIS. In \cite{zhang2022}, the authors provided a closed-form approximation for the outage probability of an RIS-aided MIMO system, and proposed a gradient descent algorithm to minimize the outage probability with statistical CSI. In addition, they characterized the achieved diversity-multiplexing tradeoff (DMT) for a finite signal-to-noise ratio (SNR) regime. Furthermore, an opportunistic rate splitting scheme was proposed in \cite{weinberger2022} for an RIS-aided two-user MISO system. In particular, by utilizing the RIS to modify the channel characteristics through an alternating CSI formulation, which was inspired by the information theoretic framework presented in \cite{chen2015}, it was shown that the achievable rate of such systems can be improved.
	
	Most of the aforementioned works assume that perfect CSI knowledge is available for the implementation of the proposed schemes. However, this assumption can be impractical, especially for a large number of RIS elements, due to the high implementation complexity or limited resources. Recently, some channel estimation protocols for RIS-aided networks have been proposed by considering several techniques, such as channel decomposition \cite{zhou2021} and discrete Fourier transform training \cite{shamasundar2022}. It can be observed that the training time of the proposed solutions increases proportionally with the size of the RIS, which could induce a large feedback overhead and compromise the expected performance gains. Moreover, several works claim that by increasing the number of elements, the performance of RIS-aided networks is enhanced. However, this highly depends on the information that is available in the system, as well as the considered communication scenario. In particular, in \cite{zhang2022} it was shown that, when only statistical CSI is available, larger RISs improve the throughput of the system, but the gain diminishes quickly by increasing the RIS elements. In addition, the authors in \cite{tyrovolas2022} considered an unmanned aerial vehicle (UAV)-based system, where an RIS is mounted to the UAV, and demonstrated that increasing the RIS size may lead to reduced data collection from the UAV.
	
	It is, therefore, an important and challenging task to provide efficient solutions, which can enhance the performance of RIS-aided networks with reduced implementation complexity. Motivated by this, in this paper, we study the performance of RIS-assisted MIMO communications under partition-based RIS schemes. Specifically, by incorporating the idea of parallel partitions in multi-hop MIMO channels \cite{yang2007}, we create a parallel channel in the time domain by intelligently adjusting different subsets of RIS elements based on a fixed reconfiguration pattern. The presented schemes have low complexity and do not require any CSI knowledge at the transmitter or for the RIS design. Specifically, the main contributions of this paper are summarized as follows:
	\begin{itemize}
		\item A general framework for RIS-aided MIMO systems is presented, where the RIS is partitioned into non-overlapping sub-surfaces, which are sequentially reconfigured to assist the communication. Based on this framework, we propose two low-complexity schemes, by considering an amplitude-based or a phase-based approach. In particular, the \textit{activate-reflect} (AR) scheme activates each sub-surface in sequential order, by changing the reflection amplitude of the corresponding elements. On the other hand, the \textit{flip-reflect} (FR) scheme creates an equivalent time-varying channel, by modifying periodically the phase shift of the elements at each sub-surface between two discrete values. The proposed schemes can be easily implemented since, in contrast to other works, channel training is not required for adjusting the reflection coefficients of the RIS elements, and a fixed reflection sequence is considered for the reconfiguration of the sub-surfaces, based on a binary state control.
		\item A complete analytical methodology for the performance of the partition-based schemes is provided. Based on the channel statistics, we derive analytical expressions which characterize the outage probability of each of the proposed schemes. Moreover, under specific practical assumptions, a more tractable methodology on the outage analysis is provided. We also study the asymptotic performance gains of the presented schemes and provide expressions for the achieved DMT. Through this analysis, we show how the idea of partitioning boosts the performance of the considered system and obtain useful insights into how some key parameters affect the performance gains.
		\item Our results demonstrate that, by employing the partition-based schemes, the outage performance of the considered system can be significantly improved compared to the conventional case, where all the RIS elements belong to a single partition and randomly rotate the signals. Moreover, both the proposed schemes achieve the same diversity gain. In particular, the achieved diversity gain is enhanced, compared to the conventional case, and this improvement is proportional to the number of partitions. Finally, it is shown that the RIS-aided MIMO system under the FR scheme always achieves the maximum multiplexing gain, while in the AR scheme the maximum multiplexing gain is obtained if certain conditions are satisfied, regarding the number of elements in each sub-surface.
	\end{itemize}
	
	The remainder of this article is organized as follows. Section \ref{sys_model} introduces the considered system model. In Section \ref{partition}, we describe the implementation of the proposed partition-based schemes and provide analytical expressions of the outage probability analysis. Section \ref{diversity} focuses on the DMT achieved by the proposed schemes. Our numerical and simulation results are presented in Section \ref{results}, and finally, some concluding remarks are stated in Section \ref{conc}.

	\textit{Notation:} Lower and upper case boldface letters denote vectors and matrices, respectively; $\mathbf{I}_{L}$ denotes the $L\times L$ identity matrix, $[\cdot]^{\dagger}$ is the conjugate transpose operator and $\left\|\cdot\right\| _F$ is the Frobenius matrix norm; $\mathbb{P}\left\lbrace X\right\rbrace $ denotes the probability of the event $X$ and $\mathbb{E}\left\lbrace X\right\rbrace$ represents the expected value of $X$; $\Gamma(\cdot)$ denotes the complete gamma function, $I_0(\cdot)$ is the modified Bessel function of the first kind of order zero and $\mathcal{Q}_{1}(\cdot,\cdot)$ is the first-order Marcum $Q$-function; $G_{p,q}^{m,n}\left( 
	\begin{smallmatrix}
		a_1,\dots,a_p\\
		b_1,\dots,b_q
	\end{smallmatrix} |
	x\right) $ is the Meijer-$G$ function \cite{gradshteyn2007} and $H_{p,q}^{m,n}\left[ \begin{smallmatrix}
		(a_i,\alpha_i,A_i)_{1,p}\\
		(b_i,\beta_i,B_i)_{1,q}
	\end{smallmatrix} |
	x\right] $ is the generalized upper incomplete Fox's $H$ function \cite{yilmaz2009}; $\Im(x)$ returns the imaginary part of $x$ and $\jmath=\sqrt{-1}$ denotes the imaginary unit; $[x]^{+}=\max(0,x)$ and $\lfloor x \rfloor =\max\{z\in\mathbb{Z}|z\leq x\}$.

\section{System Model}\label{sys_model}
	We consider a topology where a transmitter (Tx) communicates with a receiver (Rx) through the employment of an RIS, as shown in Fig. 1. The Tx and the Rx are equipped with $N$ and $L$ antennas respectively, and the RIS consists of $Q$ reflecting elements connected to a smart controller. Note that a direct link between the Tx and the Rx is not available\footnote{This assumption is not restrictive and the proposed RIS schemes can also be implemented when a direct link between the Tx and the Rx exists, but our focus is to highlight the gains provided by these schemes.} (e.g. due to high path-loss or deep shadowing) \cite{psomas2021,xu2021}. Moreover, we consider that adjacent elements are uncorrelated i.e., a half-wavelength of spacing exists between them \cite{wu2019}. Let $\mathbf{H}\in\mathbb{C}^{Q\times N}$ denote the channel matrix from the Tx to the RIS, and $\mathbf{G}\in\mathbb{C}^{L\times Q}$ the channel matrix from the RIS to the Rx. We assume a frequency-flat Rayleigh block fading channel, in which the channel coefficients remain constant during one time slot, but change independently between different time slots \cite{huang2019new,shi2022}. Moreover, each channel coefficient follows a circularly symmetric complex Gaussian distribution with zero mean and unit variance i.e., $h_{i,j}$, $g_{i,j}\sim \mathcal{CN}(0,1)$. For the rest of this paper, the channel under this topology will be referred to as the ($N,Q,L$) channel.
	
	\begin{figure}[t!]\centering
		\includegraphics[width=0.9\linewidth]{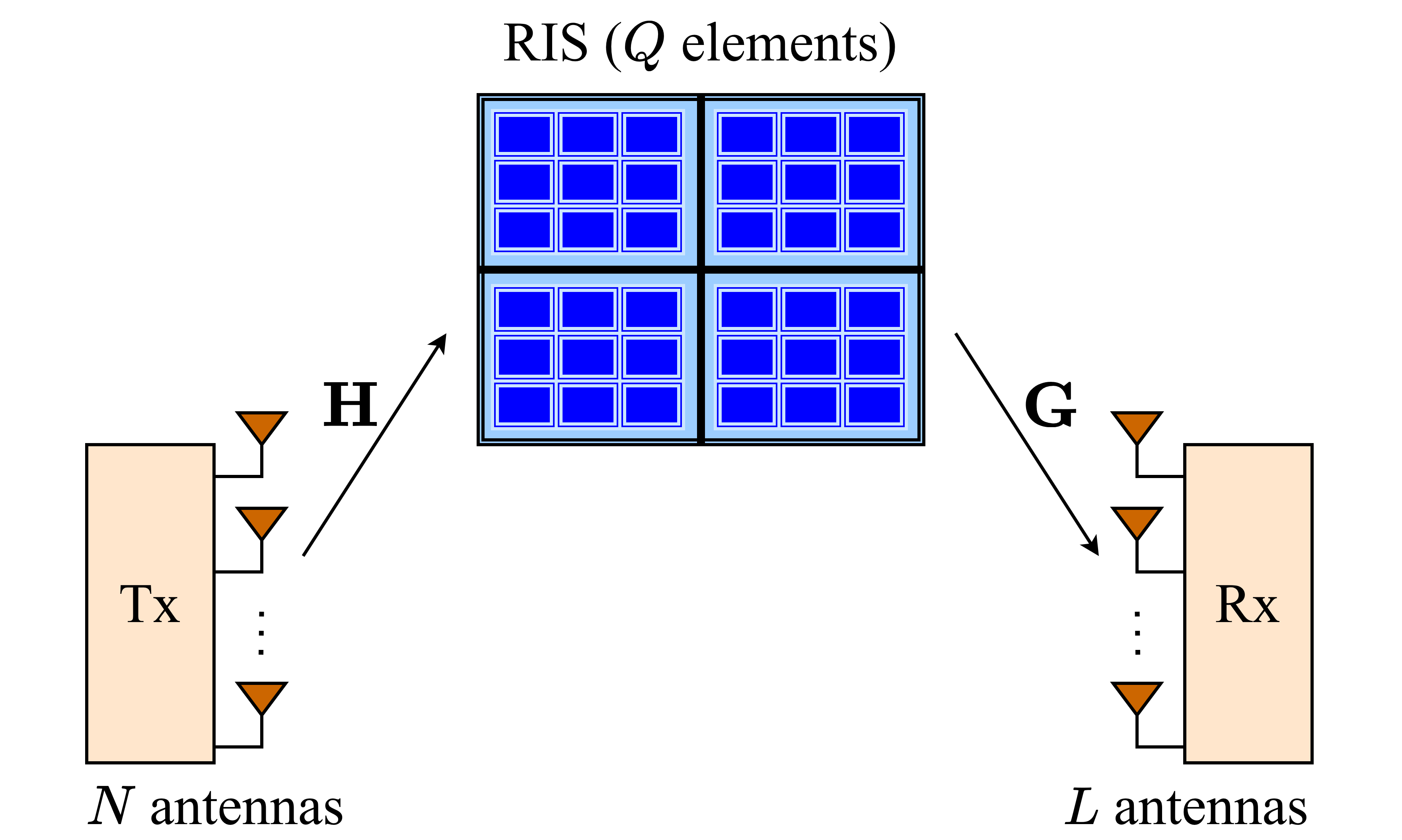}
		\caption{The considered RIS-aided channel model.}
		\label{fig:system_model}
	\end{figure}
	
	At an arbitrary time slot, the Tx sends a signal vector $\mathbf{x}\in\mathbb{C}^{N\times 1}$. We assume that CSI is perfectly known at the Rx, while CSI knowledge is not available at the Tx and the RIS. Therefore, if $\mathbf{x}$ is transmitted with a constant power $P$, the power is uniformly allocated to the $N$ transmit antennas. In the considered scenario, every time slot is divided into $K$ sub-slots of equal duration. We denote by
	\begin{equation}
		\mathbf{\Phi}_k=\text{diag}[a_{1,k}e^{\jmath\phi_{1,k}}\:\: a_{2,k}e^{\jmath\phi_{2,k}}\:\: \ldots \:\: a_{Q,k}e^{\jmath\phi_{Q,k}}],
	\end{equation}
	the diagonal reflection matrix of the RIS, where $a_{i,k}\in[0,1]$ and $\phi_{i,k}\in\left[0,2\pi\right)$ are the reflection amplitude and the phase shift of the $i$-th RIS element at the $k$-th time sub-slot, respectively. Thus, the end-to-end channel matrix during one sub-slot is written as
	\begin{equation}\label{eff_channel}
		\mathcal{H}_k=\mathbf{G}\mathbf{\Phi}_k\mathbf{H},
	\end{equation}
	and the received signal vector at the $k$-th sub-slot is given by
	\begin{equation}\label{rec_signal}
		\mathbf{y}_k=\sqrt{\dfrac{P}{N}}\mathcal{H}_k\mathbf{x}+\mathbf{n}_k,
	\end{equation}
	where $\mathbf{n}_k\in\mathbb{C}^{L\times 1}$ is the additive white Gaussian noise (AWGN) vector with entries of variance $\sigma^2$, i.e., $n_{i,k} \sim \mathcal{CN}(0,\sigma^2)$. As such, the mutual information between the Tx and the Rx over one time slot is equal to
	\begin{equation}\label{mut_inf}
		\mathcal{I}=\dfrac{1}{K}\sum_{k=1}^{K}\log_2\left[ \det\left( \mathbf{I}_{L}+\dfrac{\rho}{N}\mathcal{H}_k\mathcal{H}_k^{\dagger}\right) \right],
	\end{equation}
	where $\rho=P/\sigma^2$ is the average SNR. 
		
\section{Partition-based RIS schemes} \label{partition}

	In this section, we describe our proposed partition-based RIS schemes and analytically evaluate their performance. The proposed schemes are inspired by the idea of parallel partitions in multi-hop MIMO channels, which have been introduced in \cite{yang2007}. The main objective of these schemes is to incorporate the temporal processing into the RIS-aided channels by partitioning in a proper way the RIS elements, in order to enhance the performance of such networks. In particular, we assume that the RIS is partitioned into $K$ non-overlapping sub-surfaces $\mathcal{S}_k$, $1\leq k\leq K$ \cite{najafi2021}. This partitioning into sub-surfaces is not restricted by any specific method. Hence, for simplicity and without loss of generality, we assume that $K$ is a divisor of $Q$ so that each sub-surface has an equal number of elements\footnote{The results of the proposed scheme can be readily extended to the case where the defined sub-surfaces may contain different number of elements.} $m$ i.e., $Km=Q$. An example of the considered scenario is demonstrated in Fig. \ref{fig:system_model}, with $K=4$ and $m=9$, and where the sub-surfaces are defined by the black solid lines.
	
	At each time sub-slot $k$, $1\leq k\leq K$, the reflection configuration of the sub-surface $\mathcal{S}_k$ (i.e. the phase shifts and reflection amplitudes of the corresponding elements) is adjusted, according to the scheme that is deployed at the RIS. By using this approach, the information is transmitted at the destination by creating a set of $K$ parallel sub-channels in the time domain \cite{besser2022}. It is important to note that the number of partitions, the resulting sub-surfaces as well as the configuration sequence are determined in advance. Moreover, this information remains fixed throughout the transmission procedure and therefore can be provided to the Rx \textit{a-priori} \cite{besser2022}.

	In the following sub-sections, we present two partition-based schemes: the AR scheme, where each sub-surface is sequentially activated to assist the communication by changing the reflection amplitude of the respective elements, and the FR scheme, which modifies the phase shift of the RIS elements at each sub-surface. For these schemes, we focus on the performance evaluation in terms of outage probability. In particular, the outage probability is defined as the probability that the mutual information is below a non-negative predefined target rate of $R$ bits per channel use (bps/Hz). The general expression for the outage probability is given by
	\begin{equation}
		\Pi(R,K)=\mathbb{P}\{\mathcal{I}<R\}=\mathbb{P}\left\lbrace \sum_{k=1}^{K}\mathcal{I}_k<RK\right\rbrace,
	\end{equation}
	where $\mathcal{I}$ is given by \eqref{mut_inf} and
	\begin{equation}
		\mathcal{I}_k=\log_2\left[ \det\left( \mathbf{I}_{L}+\dfrac{\rho}{N}\mathcal{H}_k\mathcal{H}_k^{\dagger}\right) \right],
	\end{equation}
	is the mutual information of the $k$-th sub-channel. Next, we introduce the design framework that is considered for the implementation of the AR scheme. 
	
	\subsection{Activate-Reflect scheme}\label{AR}
		For the AR scheme, we assume that each RIS element can be reconfigured between two possible states, either to be turned ON or OFF \cite{mishra2019,nadeem2019}. If an element is ON, it can rotate the phase of the incident signals by a specific value, while in the OFF state, the transmitted signals can not be reflected by the element. In other words, the reflection amplitude of the ON elements is set to one (full reflection), otherwise it is set to zero (full absorption). Without loss of generality, we assume that the induced phase shift by every element is equal to zero\footnote{The resulting channel gain is statistically equivalent to the case of random phase shifts at the RIS i.e., each $\phi_{i,k}$ uniformly distributed in $\left[0,2\pi\right)$ \cite{psomas2021}.} i.e., $\phi_{i,k}=0$ $\forall i,k$ \cite{nadeem2019}.
		
		For the implementation of the AR scheme, a design framework is considered by following the steps below:
		\begin{itemize}
			\item Before transmission, all the elements of each sub-surface are set to the OFF state i.e., $a_i=0$ $\forall i$. This is considered as the default configuration of each element.
			\item Then, each sub-surface is sequentially switched to the ON state to assist the communication of the considered network. Specifically, at each sub-slot $k$, $1\leq k\leq K$, the sub-surface $\mathcal{S}_k$ is activated by turning ON only the elements that belong to the specific sub-surface. Note that any other element that does not belong to $\mathcal{S}_k$ will be reset to the default configuration. Therefore, the reflection amplitude of each RIS element at the $k$-th sub-slot is given by
			\begin{equation}\label{amplitude_AR}
				a_{i,k}=
				\begin{cases}
					1, & i\text{-th element} \in \mathcal{S}_k;\\
					0, & \text{otherwise}.
				\end{cases}
			\end{equation}
			\item The end-to-end channel is then composed of $K$ parallel sub-channels, where each sub-channel can be represented by an equivalent $(N,m,L)$ channel by considering only the activated RIS elements at each time sub-slot.
		\end{itemize}
	
		\begin{figure}[t!]\centering
			\includegraphics[width=1\linewidth]{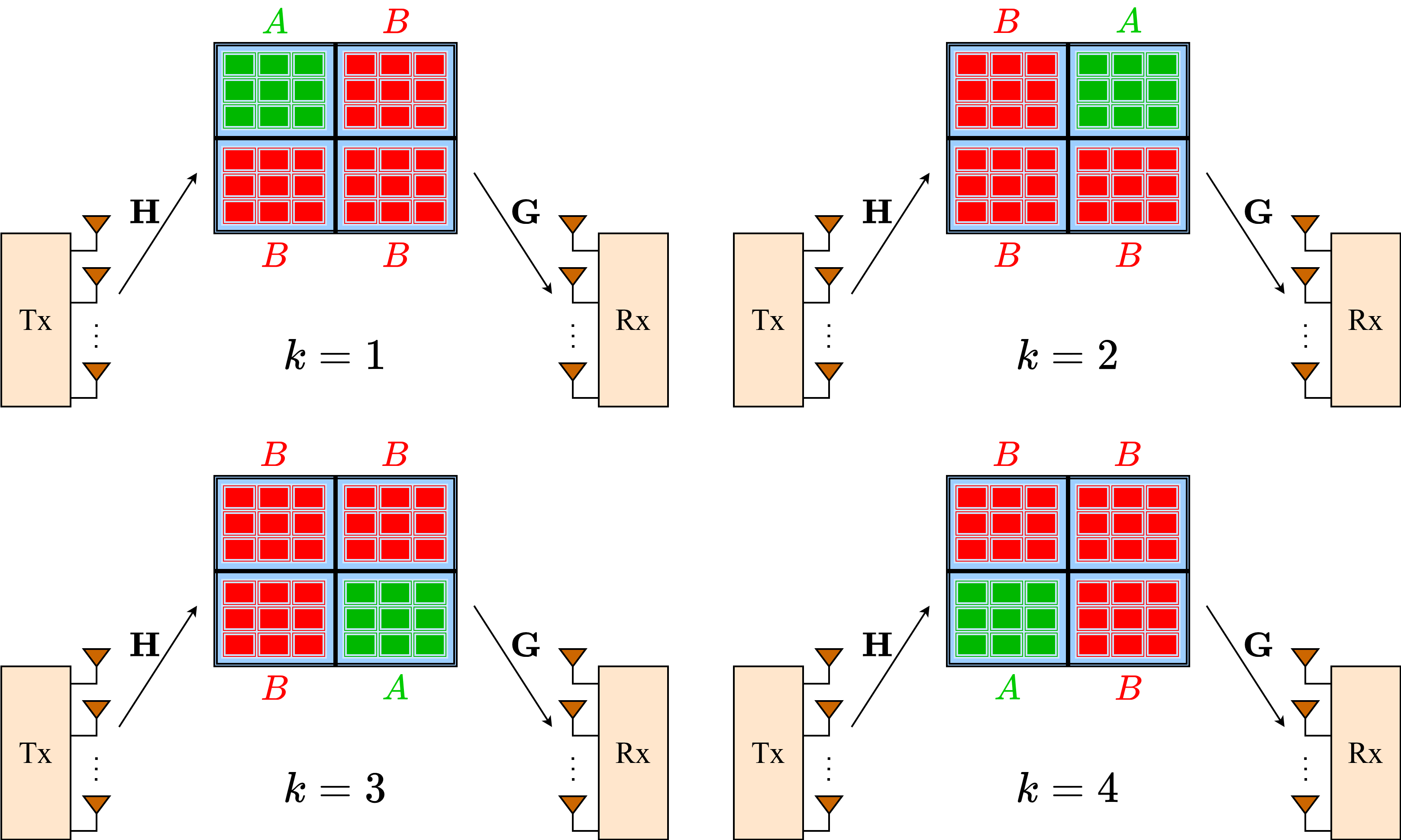}
			\caption{The different instances of the RIS configuration over one time slot for a system with $Q=36$ and $K=4$; under the AR scheme, $A=$ ON and $B=$ OFF (expression \eqref{amplitude_AR}); under the FR scheme, $A=\pi$ and $B=0$ (expression \eqref{phase_FR}).}
			\label{fig:partitions_example}
		\end{figure} 
	
		Note that the above configuration pattern corresponds to an ideal phase shift model, where the reflection amplitude of each element is independent from its corresponding phase shift; however, the AR scheme can be also employed when the reflection amplitude and the phase shift of each element are correlated \cite{abeywickrama2020}. An example of the above procedure is shown in Fig. \ref{fig:partitions_example}, where we consider the same RIS-aided channel as in Fig. \ref{fig:system_model}. Specifically, Fig. \ref{fig:partitions_example} depicts how the RIS configuration changes during an arbitrary time slot, by indicating the activated elements at each sub-slot. The proposed AR scheme could be implemented by following the approach presented in \cite{arun2020}, where each element is connected to an RF switch, which tunes the reflection amplitude of the element as either zero or one, resulting in a two-level reflection amplitude control. It is apparent that, apart from requiring no CSI knowledge, this scheme has low design complexity and is cost-effective, since we only need a two-level amplitude control for the implementation, which significantly simplifies the RIS hardware. Moreover, although the expected rate decreases, as only $m$ out of the $Q$ elements are activated at each time sub-slot, the power consumption at the RIS is reduced compared to the conventional case. Below, we present two mathematical expressions that can be used to evaluate the outage performance of the AR scheme. Specifically, 
		\begin{itemize}
			\item In Proposition \ref{pout_AR_gil}, the outage probability achieved by this scheme is derived numerically for an arbitrary number of transmit and receive antennas, by using the Gil-Pelaez inversion theorem \cite{pelaez1951}.
			\item Theorem 1 derives the outage probability for the special case of the SISO channel i.e., $N=L=1$, by using the central limit theorem (CLT), which approximates the channel $\mathcal{H}_k$ as a complex Gaussian random variable.
		\end{itemize}
		We first provide a preliminary result in the following lemma for the conventional case of $K=1$, representing the \textit{pure reflection} (PR) scheme, where all the elements randomly rotate the phase of the incident signals over an arbitrary time slot. This result will assist in the derivation of the analytical results for the presented schemes; the PR scheme will also be used as a performance benchmark for comparison purposes. 
		
		\begin{lem}
			\label{char_fun}
			The characteristic function of the mutual information of the $(N,Q,L)$ channel given in \eqref{mut_inf} under random reflections is given by
			\begin{multline}\label{eq_char_fun}
				\varphi(Q,t)=\dfrac{1}{\prod_{z=1}^{n_0}\prod_{\theta=0}^{2}\Gamma(z+\nu_\theta)}\det\bigg[ \dfrac{1}{\Gamma(-\jmath t/\ln 2)}\\
				\times G^{3,1}_{1,3}\bigg(
				\begin{array}{c}
					1\\
					-\jmath t/\ln 2,\nu_2+i,\nu_1+i+j-1
				\end{array} \bigg|
				\dfrac{N}{\rho}
				\bigg) \bigg]_{i,j}^{n_0}, 
			\end{multline}
			where $(n_0,n_1,n_2)$ is the ordered version of the $(N,Q,L)$ channel and $\nu_i\triangleq n_i-n_0$, $0\leq i\leq 2$.
		\end{lem}
		\begin{proof}
			See Appendix \ref{proofA}.
		\end{proof}
		Next, by using the above lemma, the outage probability of the AR scheme is derived as follows.
	
		\begin{prop}\label{pout_AR_gil}
			The outage probability of the AR scheme is given by
			\begin{equation}\label{out_partition}
				\Pi_{\AR}(R,K,m)=\dfrac{1}{2}-\dfrac{1}{\pi}\int_{0}^{\infty}\dfrac{1}{t}\Im \biggl\{\left[e^{-\jmath tR}\varphi(m,t)\right]^K\biggr\}dt,
			\end{equation}
			where $\varphi(m,t)$ is the characteristic function of the mutual information of the ($N,m,L$) channel given by Lemma \ref{char_fun}.
		\end{prop}
		\noindent Recall that, at each time sub-slot, only $m$ elements of the RIS are activated. The above result can be easily derived by considering that the random variables $\mathcal{I}_k$ are independent and the characteristic function of the sum $\sum_{k=1}^{K}\mathcal{I}_k$ is calculated by the product of the characteristic functions of each $\mathcal{I}_k$. It is also clear that, for $K=1$, the above proposition provides a numerical expression for the outage probability of the ($N,Q,L$) channel under the conventional PR scheme. We next provide an approximation of the outage probability for the RIS-aided SISO channel. 
		
		\begin{thm}\label{pout_AR_CLT}
			The outage probability achieved by the RIS-assisted SISO channel employing the AR scheme, under the CLT, is approximated by
			\begin{multline}\label{pout_approx_AR}
				\Pi_{\AR}(R,K,m)\approx 1-\exp\left( \dfrac{K}{\rho m}\right)\\
				 \times H^{K+1,0}_{1,K+1}\left[ 
				\begin{array}{c}
					\left( 1,1,0\right) \\
					\left( 0,1,0\right) ,\left(1,1,\dfrac{1}{\rho m}\right) ,\dots,\left( 1,1,\dfrac{1}{\rho m}\right) 
				\end{array} \Bigg|\Theta\right] ,
			\end{multline}
			where $\Theta=(2^{R}/\rho m)^{K}$.
		\end{thm}
	
		\begin{proof}
			See Appendix \ref{proofB}.
		\end{proof}
	
		\noindent The above expression can be easily evaluated by using computational software tools, such as Matlab or Mathematica \cite{yilmaz2009}, while in Section V, we show that the provided approximation becomes very tight, even for a small number of elements in each sub-surface. The above approximation can be also considered for the PR scheme by setting $K=1$. In this case, the expression is simplified to
		\begin{equation}\label{PR_approx}
			\Pi_{\PR}(R)\approx 1-\exp\left( -\dfrac{2^R-1}{\rho Q}\right) .
		\end{equation}

	\subsection{Flip-Reflect scheme}\label{FR}
		We now consider a partition-based scheme, where the reconfiguration of each sub-surface $\mathcal{S}_k$ occurs on the phase shift of the selected elements. For the FR scheme, in particular, all the elements of the RIS are always turned ON and are thus able to reflect the incident signals i.e., $a_{i,k}=1$ $\forall i,k$. At every sub-slot $k$, the RIS controller can modify the induced phase shift of each element between $0$ or $\pi$. For the implementation of the FR scheme, we adopt a similar framework as for the AR scheme. Specifically, we consider the following procedure:
		\begin{itemize}
			\item Initially, the phase shift of all the RIS elements is set to zero. This setting is considered as the default configuration of each element for the FR scheme.
			\item Regarding the number of partitions, we consider the cases\footnote{By considering the above cases for the flipping pattern, we ensure that the RIS reflection matrices are linearly independent, so that the maximum performance gains can be achieved.} $K=2$ and $K>2$. If $K=2$, at the first time sub-slot all the elements remain at the default configuration and simply reflect the signals i.e., $\phi_{i,1}=0$ $\forall i$. At the second sub-slot, the elements of $\mathcal{S}_2$ are reconfigured by setting their phase shift into $\pi$. On the other hand, for $K>2$, at the $k$-th sub-slot, the sub-surface $\mathcal{S}_k$ flips the signals by changing the phase shift of the corresponding elements into $\pi$. As such, the phase shift of each RIS element at the $k$-th sub-slot is equal to
			\begin{equation}\label{phase_FR}
				\phi_{i,k}=
				\begin{cases}
					\pi, & (K>2\cup (K=2\cap k>1))\cap i \in \mathcal{S}_k;\\
					0, & \text{otherwise}.
				\end{cases}
			\end{equation}
			\item Therefore, the ($N,Q,L$) channel under the FR scheme is a time-varying channel consisting of $K$ parallel sub-channels, where for each sub-channel a different flipping matrix $\Phi_k$ is used \cite{yang2010new,pedarsani2015new}.
		\end{itemize}
		\noindent Similarly to the AR scheme, we assume that the reflection amplitude and phase shift of each element can be independently modified, while the proposed design framework can be generalized to the case where the reflection amplitude and the phase shift of each element are coupled \cite{abeywickrama2020}. Fig. \ref{fig:partitions_example} depicts an example of the above procedure, by indicating the value of the phase shift for the elements of every sub-surface at each time sub-slot. The 1-bit phase-shift control presented in this framework could be implemented based on the binary programmable metasurface fabricated in \cite{amri2020}, where each RIS element is connected to a PIN diode that can be switched between two states, resulting in a phase shift difference of $\pi$. Therefore, similar to the AR scheme, this scheme has low complexity in terms of RIS hardware requirements, as the proposed configuration patterns can be obtained with low-cost binary-state elements.
		
		Based on the presented framework, we can now analyze the performance of the FR scheme in terms of outage probability. Note that, since all the RIS elements are always activated in this scheme, the instantaneous channel gains between different sub-slots are time-correlated. Thus, in this case, the derivation of the outage probability becomes challenging. As such, we present two approximations that can sufficiently describe the performance of the proposed scheme. In particular,
		\begin{itemize}
			\item In Theorem 2, we derive a numerical expression for the outage probability of an RIS-aided SISO channel by approximating the cascaded channels $\mathcal{H}_k$ with time-correlated Rayleigh fading channels.
			\item For the general case of RIS-aided MIMO channels, we provide a lower bound by assuming that the channels $\mathcal{H}_k$ are mutually independent.
		\end{itemize}
	
		In what follows, we analytically evaluate the correlation coefficient of the channel gains over different sub-slots for the SISO case i.e., for $N=L=1$, by using the Pearson correlation formula \cite{papoulis2002}.
	
		\begin{lem}\label{corr_coeff}
			The correlation coefficient of the channel gains over different sub-slots $k\neq l$ under the FR scheme is given by
			\begin{equation}\label{corr_expr}
				\zeta=\dfrac{(Q-2b)^2+2Q}{Q(Q+2)},
			\end{equation}
			where
			\begin{equation}\label{b_value}
				b=
				\begin{cases}
					m, & K=2;\\
					2m, & K>2.
				\end{cases}
			\end{equation}
		\end{lem}

		\begin{proof}
			See Appendix \ref{proofC}.
		\end{proof}

		\begin{figure}[t!]\centering
			\includegraphics[width=1\linewidth]{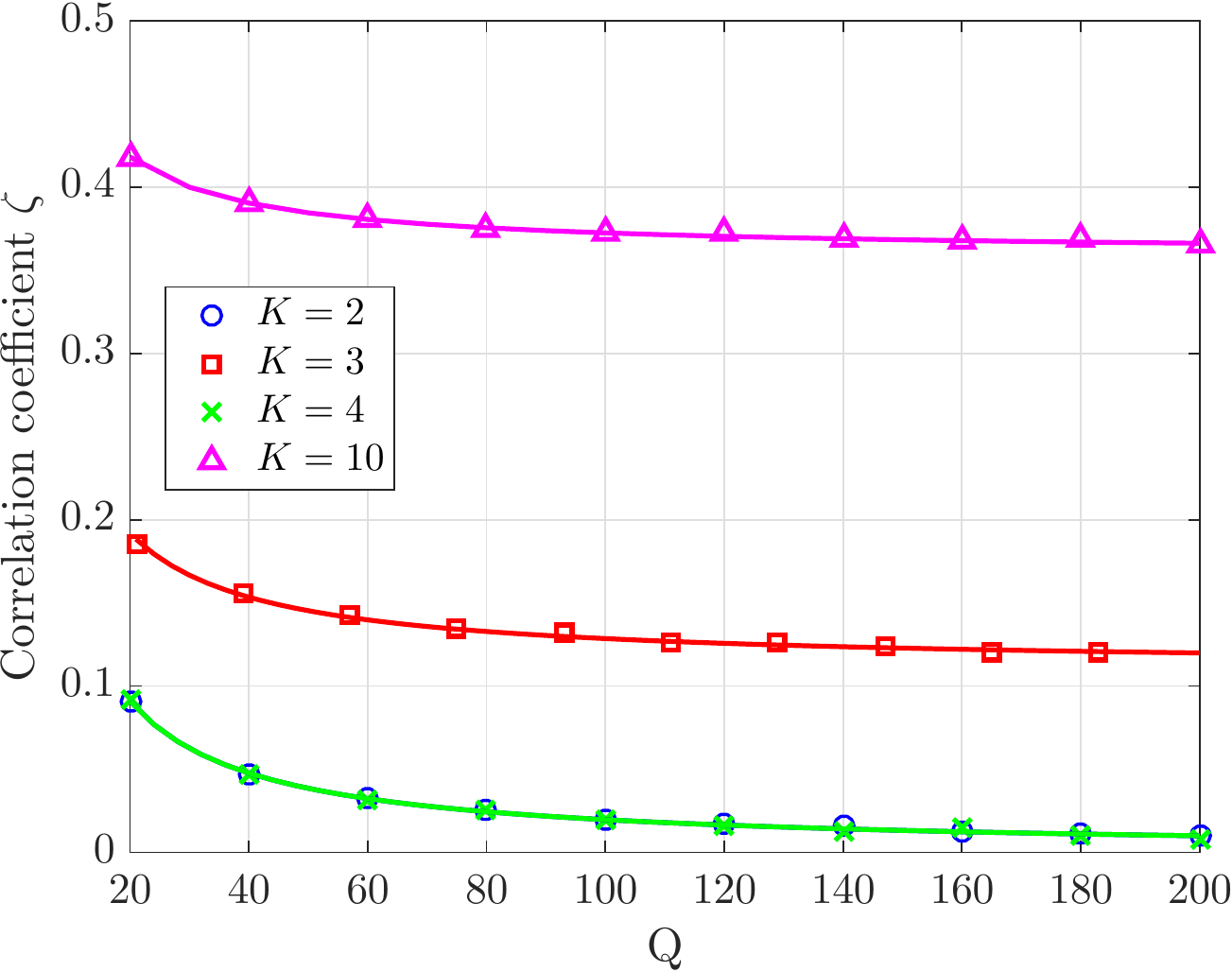}
			\caption{Correlation coefficient $\zeta$ versus number of elements $Q$; the theoretical results are depicted with lines and the simulation results with markers.}
			\label{fig:corr_coeff_fig}
		\end{figure}

		\noindent It can be easily verified that when $Q\rightarrow\infty$, the correlation coefficient converges to
		\begin{equation}
			\zeta\rightarrow
			\begin{cases}
				0, & K=2;\\
				1-\dfrac{8(K-2)}{K^2}, & K>2.
			\end{cases}
		\end{equation}
		From the above expression, we can deduce that, the correlation between the channel gains is relatively small for small partition sizes. In particular, we observe that when $K=2$ or $4$ the correlation coefficient converges to zero. On the other hand, as $K$ increases, $\zeta$ converges to one i.e., full correlation. These results are also depicted in Fig. \ref{fig:corr_coeff_fig}.
	
		We can now focus on the derivation of the approximated expression of the outage probability for the RIS-aided SISO networks. Since the considered channels are time-correlated, we adopt a correlated Rayleigh fading channel model to approximate their distribution, and incorporate the correlation coefficients provided in the above lemma. Specifically, the approximated channel can be written as \cite{beaulieu2011}
		\begin{equation}\label{corr_Rayleigh}
			\mathcal{\tilde{H}}_k=\begin{cases}
				\sigma_kX_1, & k=1;\\
				\sigma_k(\sqrt{1-\zeta}X_k+\sqrt{\zeta} X_1), & 2\leq k\leq K,
			\end{cases}
		\end{equation}
		where $X_k$, $1\leq k\leq K$, are independent complex Gaussian random variables with zero mean and unit variance and $\sigma_k^2=\mathbb{E}\left\lbrace \left| \mathcal{H}_{k}\right| ^{2}\right\rbrace=Q$. Let $W_k=\left| \mathcal{\tilde{H}}_k\right| ^2$ denote the channel gain at the $k$-th sub-slot. Based on the definition of $\mathcal{\tilde{H}}_k$, we deduce that $W_1$ is exponentially distributed with parameter $1/Q$, where the associated probability density function (PDF) is given by
		\begin{equation}\label{pdf_W_1}
			f_{W_1}(x)=\dfrac{1}{Q}\exp\left(-\dfrac{x}{Q}\right).
		\end{equation} 
		Conditioned on $W_1$, the remaining terms $W_k$, $2\leq k \leq K$, follow an independent non-central chi-squared distribution with two degrees of freedom. Therefore, the conditional PDF of $W_k$, given $W_1=y$, is written as \cite[Theorem~1.3.4]{muirhead1982}
		\begin{equation}\label{cond_pdf_W_k}
			f_{W_k|W_1}(x|y)=\dfrac{\exp\left( -\dfrac{x+\zeta y}{\Omega}\right)}{\Omega} I_0\left( 2\dfrac{\sqrt{\zeta xy}}{\Omega}\right),
		\end{equation}
		where $\Omega\triangleq Q \left( 1-\zeta\right)$. As such, the approximated outage expression is given as follows.
		\begin{thm}\label{approx_FR_thm}
			The outage probability of the SISO channel, under the FR scheme, is approximated by
			\begin{multline}
				\Pi_{\FR}(R,K,m)\approx \left( \dfrac{1}{\rho}\right) ^{K-1}\int_{1}^{c_1}\int_{1}^{c_2}\cdots\int_{1}^{c_{K-1}}\\
				\times\left[ 1-\mathcal{Q}_{1}\left( \sqrt{\dfrac{2\zeta(\tau_1-1)}{\rho\Omega}},\sqrt{2\Theta}\right) \right]f_{W_1}\left( \dfrac{\tau_1-1}{\rho}\right)\\
				\times\prod_{k=2}^{K-1}f_{W_k|W_1}\left( \dfrac{\tau_k-1}{\rho}\middle|\dfrac{\tau_1-1}{\rho}\right)d\tau_{K-1}\cdots d\tau_1,
			\end{multline}
			where $c_i\triangleq2^{RK}/\prod_{k=1}^{i-1}\tau_k$, $1\leq i \leq K-1$	and
			\begin{equation}
				\Theta\triangleq\dfrac{1}{\rho\Omega}\left(\dfrac{2^{RK}}{\prod_{k=1}^{K-1}\tau_k}-1 \right) .
			\end{equation}
		\end{thm}
		\begin{proof}
			See Appendix \ref{proofD}.
		\end{proof}
		\noindent We show that the presented approach provides a very tight approximation of the outage performance under the FR scheme, and can adequately describe the system's behavior, even for small values of $Q$.
		
		Finally, for the general case where $N,L\geq1$, we point out some remarks on the performance of the considered RIS-aided networks under the FR scheme. For this case, specifically, we provide a lower bound on the outage probability by assuming that all the resulting parallel channel matrices $\mathcal{H}_k$ are mutually independent. Under this assumption, the ($N,Q,L$) channel, by using the FR scheme, achieves the same performance as the ($N,KQ,L$) channel under the AR scheme with the same partition size $K$. The outage probability of the FR scheme is therefore lower bounded by
		\begin{equation}\label{FR_lower_bound}
			\Pi_{\FR}(R,K,m)\geq\Pi_{\AR}(R,K,Q),
		\end{equation}
		where $\Pi_{\AR}(R,K,Q)$ is given by \eqref{out_partition}. The above performance bound becomes tight when $m$ (and therefore $Q$) is sufficiently large and the number of partitions $K$ is relatively small, since the correlation of the channel gains between different sub-slots remains low. Moreover, based on this result, we can conclude that the AR scheme requires the employment of an RIS with $K$ times more elements to outperform the outage probability achieved by the FR scheme.

\section{Diversity-multiplexing tradeoff analysis}\label{diversity}
	We now turn our attention to the DMT achieved by the proposed schemes, which gives the performance limit for uncoded transmission. In general, a scheme achieves multiplexing gain $r$ and diversity gain $d(r)$ if the target data rate $R(\rho)\sim r\log\rho$ and the outage probability of the scheme $\Pi(\rho)$ satisfy the conditions \cite{zheng2003}
	\begin{equation*} 
		\lim_{\rho\rightarrow \infty}\dfrac{R(\rho)}{\log \rho}=r,
	\end{equation*}
	and
	\begin{equation}
		\lim_{\rho\rightarrow \infty}-\dfrac{\log \Pi(\rho)}{\log \rho}=d(r).
	\end{equation}
	By following the definitions of \cite{yang2007}, if the achieved DMT of the end-to-end channels for any two schemes is the same, then these channels are said to be \textit{DMT-equivalent}. Let ($n_0, n_1, n_2$) be the ordered version of the ($N,Q,L$) channel, with $n_0\leq n_1\leq n_2$. An ordered ($l_0, l_1, l_2$) channel is a \textit{vertical reduction} of the considered channel, if both are DMT-equivalent and satisfy the condition $l_i\leq n_i$ $\forall i$. Finally, according to the information theoretic cut-set bound \cite{cover1991}, the DMT achieved by any RIS scheme is upper bounded as
	\begin{equation}
		d_{(N,Q,L)}(r)\leq\min\{d_{(N,Q)}(r),d_{(Q,L)}(r)\},
	\end{equation}
	where the maximum diversity and multiplexing gain are given by
	\begin{equation}\label{max_d_r} 
		d_{\max}=\min\{N,L\}\times Q,
	\end{equation}
	and
	\begin{equation}\label{max_r} 
		r_{\max}=\min\{N,Q,L\}.
	\end{equation}
	
	Before deriving the DMT achieved by the presented schemes, we need to provide the DMT expression for the conventional PR scheme, which is given below. The $(N,Q,L)$ channel under the PR scheme is DMT-equivalent to the Rayleigh product channel. Therefore, the achieved DMT is a piecewise-linear function defined by the points $(r,d_{\PR}(r)),r=0,...,n_0$ with \cite[Theorem~2]{yang2011}
	\begin{equation}\label{dmt_pure}
		d_{(N,Q,L)}^{\PR}(r)=(n_0-r)(n_1-r)-\bigg\lfloor\dfrac{[(n_0+n_1-n_2-r)^+]^2}{4}\bigg\rfloor.
	\end{equation}
	\noindent Based on the above expression, some important remarks can be extracted about the DMT performance of an RIS-aided network operating under the PR scheme. Specifically,
	\begin{itemize}
		\item The DMT achieved by the PR scheme depends only on the ordered version of the $(N,Q,L)$ channel.
		\item The DMT of the $(N,Q,L)$ channel under the PR scheme is limited by the DMT performance of the $N\times L$ channel if the number of RIS elements satisfies the condition $Q\geq N+L-1$. This can be easily proven since, under this condition, the last term of \eqref{dmt_pure} is equal to zero.
		\item The $(N,Q,L)$ channel can be vertically reduced to any $(N,\tilde{Q},L)$ channel, with $Q>\tilde{Q}\geq Q_{min}=N+L-1$, and still achieve the same DMT.
	\end{itemize}
	
	It is therefore easily observed that the PR scheme is suboptimal in terms of diversity gain, compared to the maximum diversity gain associated with the theoretical cut-set bound. In typical RIS-assisted networks, the number of RIS elements is usually much larger than the number of transmit and receive antennas. Therefore, the DMT achieved by this scheme is limited by the ``bottleneck'' ordered channel. For $Q>Q_{min}$ the performance of the $(N,Q,L)$ channel can only be improved in terms of coding gain. However, the achieved coding gain decreases as the number of elements increases. This can be easily proven for the RIS-aided SISO channel, by using the approximated expression in \eqref{PR_approx}. In this case, by using the definition of \cite{tse2005}, the coding gain is given by
	\begin{align}
		\mathcal{G}&=\lim_{\rho\rightarrow \infty} \Pi_{\PR}(R)\rho^{d_{(1,Q,1)}(0)}\nonumber\\
		&=\lim_{\rho\rightarrow \infty} \left[ 1-\exp\left( -\dfrac{2^R-1}{\rho Q}\right)\right] \rho\approx\dfrac{2^R-1}{Q},		
	\end{align}
	where we used the approximation $\exp(-x)\approx1-x$ for $x\rightarrow0$. Apparently, when $Q\rightarrow\infty$, the coding gain converges to zero. We next evaluate the DMT of the proposed AR scheme.
		
	\begin{cor}\label{dmt_AR}
		The DMT achieved by the ($N,Q,L$) channel under the AR scheme is equal to
		\begin{equation}
			d_{(N,Q,L)}^{\AR}(r)=Kd_{(N,m,L)}^{\PR}(r),
		\end{equation}
		where $K$ is the number of sub-surfaces with $m$ elements.
	\end{cor}
	\noindent The above result is derived by considering that in the proposed scheme, the end-to-end channel consists of $K$ independent parallel sub-channels, and each sub-channel achieves the same DMT given in \eqref{dmt_pure} with $Q=m$. It is apparent that the AR scheme can significantly enhance the performance of the $(N,Q,L)$ channel in terms of diversity gain compared to the PR scheme, especially for a large number of RIS elements. In particular, this scheme can achieve both the maximum diversity and multiplexing gain associated with the cut-set bound, if the RIS is divided into $K$ sub-surfaces of $m$ elements with 
	\begin{equation}\label{cond_dmt}
		\min \{N,L\}\leq m\leq \left| N-L\right| +1,
	\end{equation}
	and $m$ is a divisor\footnote{In the generalized case of an RIS divided into sub-surfaces with different number of elements, $m_i$, $1\leq i\leq K$, just needs to be an integer number within the defined region (see \eqref{cond_dmt}).} of $Q$. If the above bounds can not be simultaneously satisfied, then the AR scheme can achieve either $d_{\max}$ or $r_{\max}$ given by \eqref{max_d_r} and \eqref{max_r}, respectively.
		
	It can be seen that, despite the improvement achieved in diversity gain by the employment of the AR scheme, the end-to-end channel suffers from rate-deficiency, while the maximum multiplexing gain is no longer guaranteed. On the other hand, this issue is resolved with the FR scheme. In this case, the exact DMT is difficult to be obtained for the FR scheme, so we provide a lower bound instead, which is given in the following proposition.
	\begin{figure}[t!]\centering
		\includegraphics[width=1\linewidth]{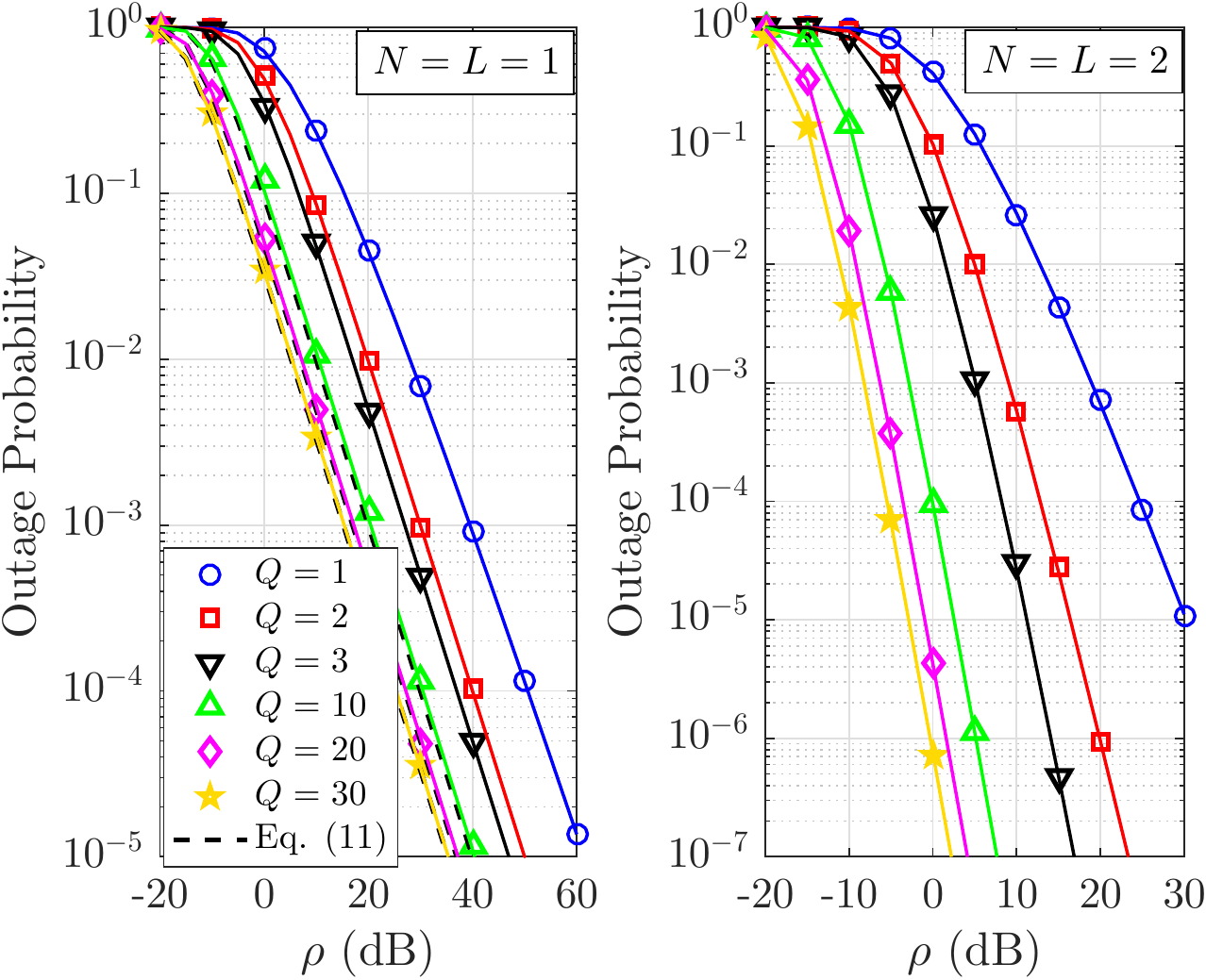}
		\caption{Outage probability versus average SNR for the PR scheme ($K=1$).}
		\label{fig:pure_refl_snr}
	\end{figure}
	\begin{prop}\label{prop_dmt_FR}
		The DMT achieved by the ($N,Q,L$) channel under the FR scheme is lower bounded by
		\begin{equation}\label{dmt_FR_LB}
			d_{(N,Q,L)}^{\FR}(r)\geq \max{\left\lbrace d_{(N,Q,L)}^{\AR}(r),d_{(N,Q,L)}^{\PR}(r)\right\rbrace }.
		\end{equation}
	\end{prop}
	\begin{proof}
		See Appendix \ref{proofE}.
	\end{proof}
	
	We can observe that the FR scheme is superior to both the AR and the PR schemes. Based on the framework of the FR scheme, all the RIS elements are always activated and assist the communication, so the maximum multiplexing gain is ensured. At the same time, through the temporal processing pattern presented in Section \ref{FR}, we can increase the achieved diversity gain. Therefore, while the AR and the PR schemes can achieve the maximum diversity and multiplexing gain respectively, the FR scheme achieves both extremes. Moreover, the ($N,Q,L$) channel under the FR scheme can achieve both $d_{\max}$ and $r_{\max}$, by considering only the upper bound condition of \eqref{cond_dmt}.

\section{Numerical Results} \label{results}

	We provide numerical results to demonstrate the performance of the presented partition-based RIS schemes and validate our theoretical analysis. For the simulations, we consider that the rate threshold is equal to $R=1$ bps/Hz throughout the simulations, while the variance of each entry at the AWGN vector is normalized to $\sigma^{2}=1$. Note that, the considered set of parameter values is used for the sake of presentation. A different set of values will affect the performance but will lead to similar observations. Moreover, in the following results, the proposed schemes are compared with the conventional PR scheme i.e., the case of $K=1$. Unless otherwise stated, the analytical results are illustrated with lines (solid, dashed or dotted) and the simulation results with markers.
	
	\begin{figure}[t!]
		\centering
		\includegraphics[width=0.95\linewidth]{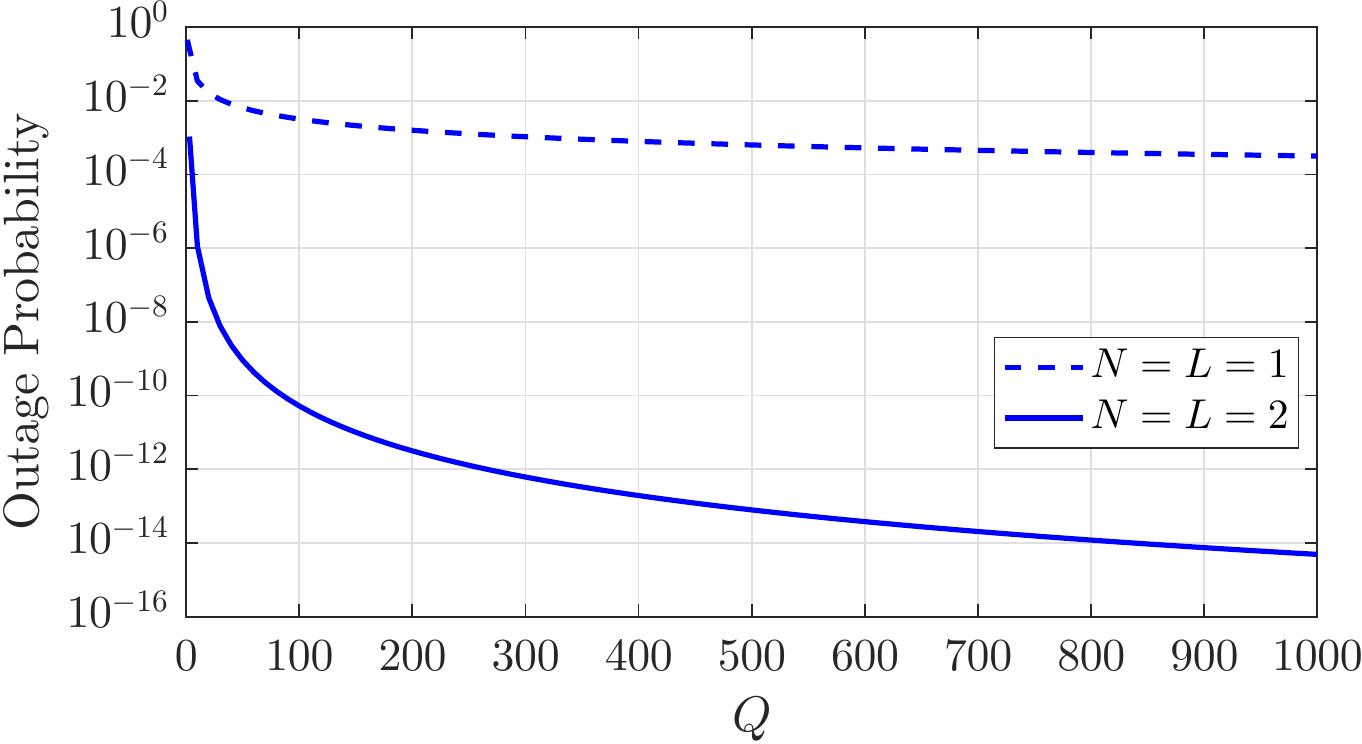}
		\caption{Outage probability versus $Q$ under the PR scheme for $\rho=5$ dB.}
		\label{fig:pure_refl_Q}
	\end{figure}	

	Fig. \ref{fig:pure_refl_snr} illustrates the system's outage probability with respect to the average SNR under the PR scheme for different values of RIS elements. We show the results of two different cases: the first case considers a SISO network setting i.e., $N=L=1$ antenna, while in the second case we have a MIMO setting with $N=L=2$ antennas. We can see that the outage performance is improved, by increasing the number of RIS elements. Specifically, for the RIS-aided SISO channel we observe that for all values of $Q$ a diversity order equal to one is achieved and the system's performance is improved in terms of coding gain. On the other hand, the performance in the MIMO setting is enhanced in terms of diversity gain until the RIS elements reach the value of $Q_{min}=N+L-1=3$, where $d(0)=NL=4$. After this threshold, the diversity order remains the same as we increase $Q$, so the outage probability can be only improved in terms of coding gain. This observation is in accordance with our remark that the asymptotic performance of the PR scheme is limited by the $N\times L$ channel. Finally, the presented results validate the accuracy of our theoretical analysis. Specifically, we observe that in both cases the simulation results (markers) perfectly match with the theoretical values (solid lines) of Proposition \ref{pout_AR_gil} for $K=1$. In addition, the expression in \eqref{PR_approx} provides an exceptional approximation (dashed lines) of the achieved outage probability, even for small values of $Q$, while as $Q$ increases the approximation becomes tighter.
	
	\begin{figure}[t!]
		\centering
		\includegraphics[width=1\linewidth]{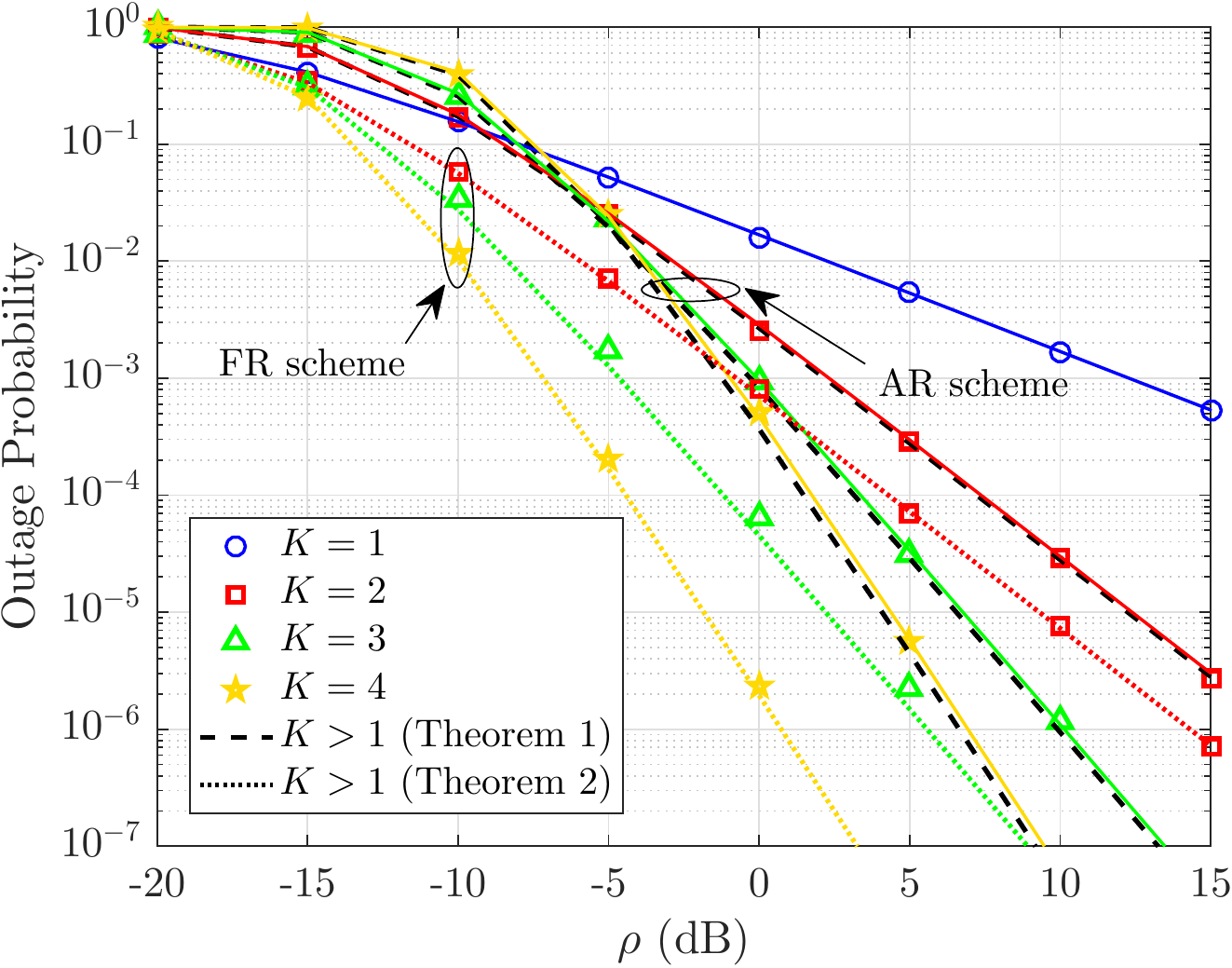}
		\caption{Outage probability versus average SNR for the partition-based schemes; $N=L=1$ antenna and $Q=60$.}
		\label{fig:partition_siso}
	\end{figure}
	
	As noted previously in Fig. \ref{fig:pure_refl_snr}, for $Q>Q_{min}$ the outage performance is improved only in terms of coding gain. However, this gain gradually diminishes as the number of elements increases. This observation is more evident in Fig. \ref{fig:pure_refl_Q}, where the outage probability is shown against the number of elements for the same channel settings at $\rho=5$ dB, and is consistent with the conclusions derived in \cite{zhang2022}. We can therefore deduce that under the PR scheme, employing a larger RIS does not necessarily provide significant gains to the outage probability of the RIS-aided network.

	In Figs. \ref{fig:partition_siso} and \ref{fig:partition_mimo}, we present the achieved outage probability of the considered RIS-aided system under the proposed partition-based schemes (AR and FR schemes), for a network topology with $N=L=1$ and $2$ antennas, respectively, $Q=60$ elements and different partition sizes. As expected, by increasing the number of transmit and receive antennas, the outage probability as well as the achieved diversity gain are improved. Again, the theoretical results regarding the AR scheme (solid lines) are in agreement with the simulations (markers) in both figures, which validates our analysis. In Fig. \ref{fig:partition_siso}, we also show the approximations of Theorem \ref{pout_AR_CLT} and Theorem \ref{approx_FR_thm} for the two schemes. It is observed that, in both cases, the derived expressions follow the behavior of the actual performance of the system and provide a very tight approximation. Furthermore, in Fig. \ref{fig:partition_mimo}, we include the lower bound of the FR scheme given in \eqref{FR_lower_bound}, assuming mutually independent channel matrices $\mathcal{H}_k$. We can see that, for $Q\gg K$ the lower bound is close to the performance achieved by the FR scheme, which validates its consideration.
	
	\begin{figure}[t!]
		\centering
		\includegraphics[width=1\linewidth]{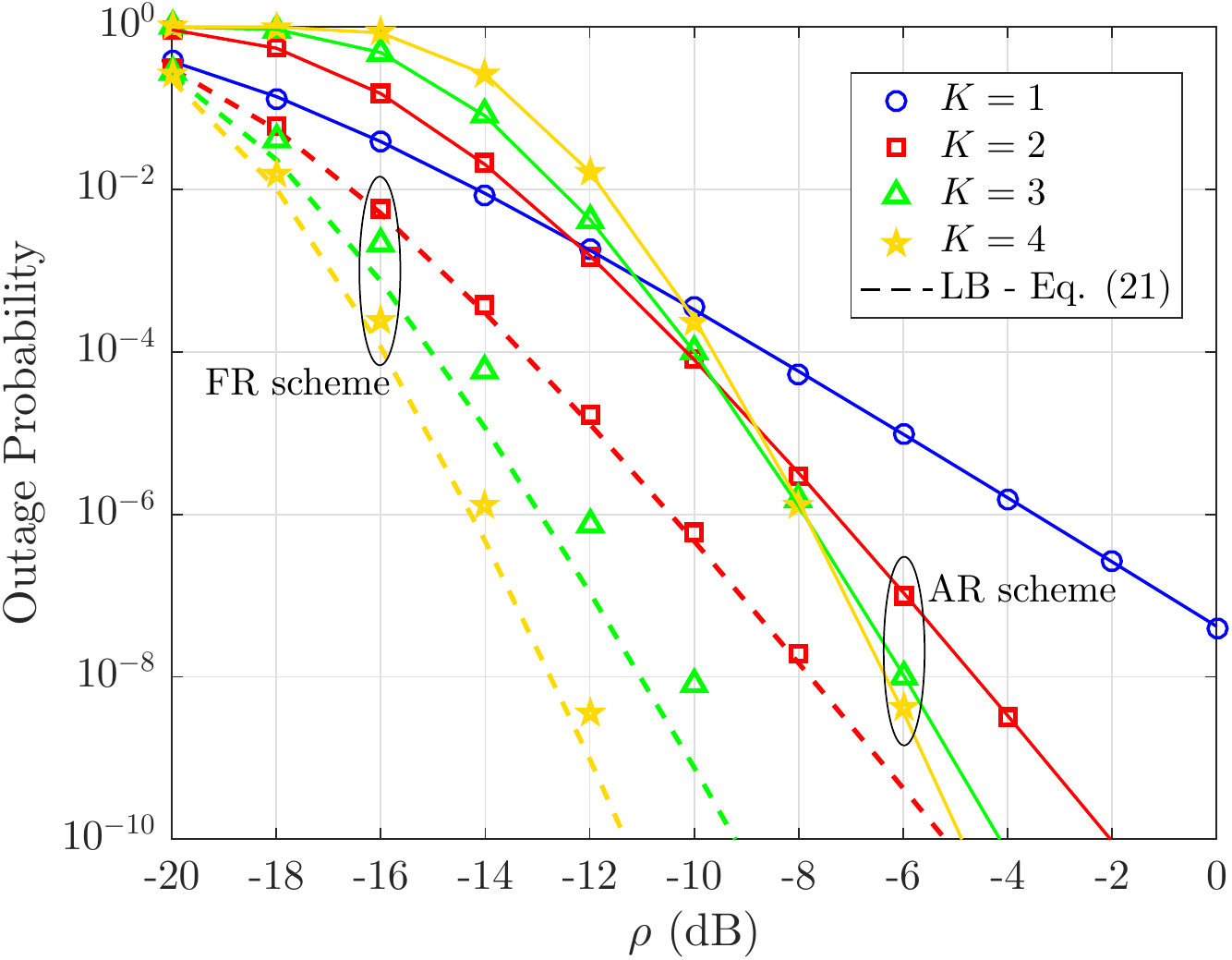}
		\caption{Outage probability versus average SNR for the partition-based schemes; $N=L=2$ antennas and $Q=60$.}
		\label{fig:partition_mimo}
	\end{figure}

	In both figures, the main observation is that for high SNR values, both the AR scheme and the FR scheme outperform the conventional PR scheme. Specifically, by increasing the number of partitions, the asymptotic performance of the presented schemes is improved in terms of diversity gain. We can also clearly see that the proposed schemes achieve the same diversity order, which is in accordance with our results in Section \ref{diversity}. However, the FR scheme outperforms the AR scheme in terms of coding gain. This is due to the fact that in the FR scheme, all the RIS elements are activated at each sub-slot, while in the AR scheme only the elements of the selected sub-surface are turned on. Moreover, we observe that, as the number of partitions increases, the gain achieved by the FR scheme over the AR scheme increases as well. In particular, in both figures, it can be seen that partitioning the RIS into $4$ sub-surfaces almost doubles the coding gain between the two schemes, compared to the case of $K=2$. On the contrary, in the low SNR regime only the FR scheme outperforms the conventional case. This is expected since, according to the framework of the AR scheme, at each time sub-slot the spectral efficiency of the considered model is deteriorated due to the reduced duration of each time sub-slot and the decreased number of RIS elements selected at each sub-surface.
	
	\begin{figure}[t!]
		\centering
		\includegraphics[width=1\linewidth]{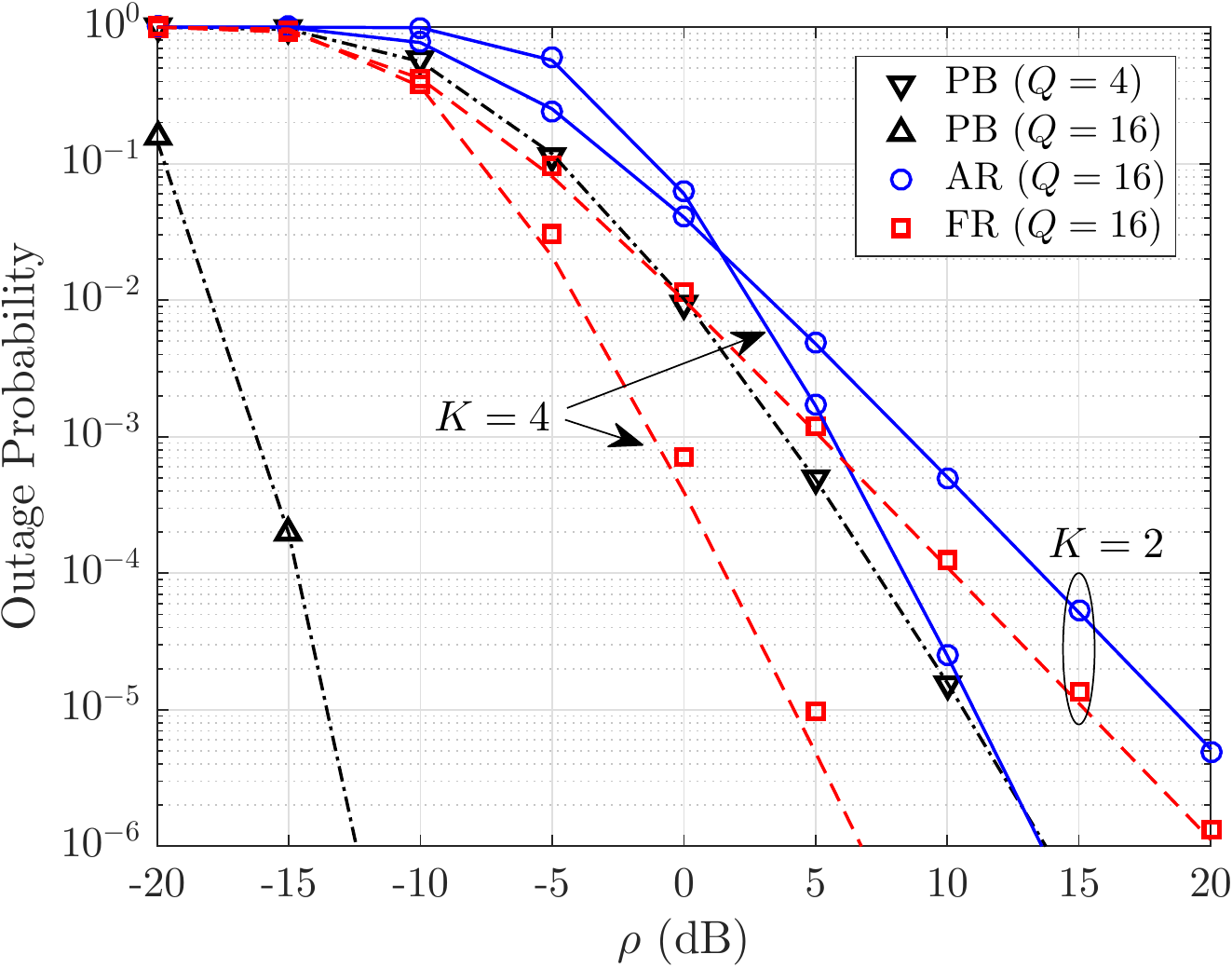}
		\caption{Outage probability comparison of the partition-based schemes and the PB case for a network topology with $N=L=1$.}
		\label{fig:comp_partitions_PB}
	\end{figure}

	In Fig. \ref{fig:comp_partitions_PB}, we compare the outage performance of the proposed schemes with the case of passive beamforming (PB) at the RIS, by considering an RIS-aided SISO channel. In the PB case, the Tx estimates the cascaded channel through a channel estimation method, e.g. \cite{mishra2019}, and aligns the phase shifts of the RIS elements i.e., $\phi_{i}=-\angle h_ig_i$, with $1\leq i\leq Q$. We consider an ideal scenario where no errors are derived from the channel estimation, so we have perfect CSI knowledge. The outage probability achieved by the PB is obtained in \cite[Theorem~3]{psomas2021}. Fig. \ref{fig:comp_partitions_PB} depicts the outage probability of the PB for $Q=4$ and $16$ elements, and the outage probability of the AR and FR schemes for $Q=16$ and different partition sizes. As expected, the PB case outperforms the proposed schemes. On the other hand, it can be seen that, by increasing the number of partitions at the RIS, the difference between the outage performance of the proposed schemes and the PB decreases. Moreover, a fair comparison in terms of power consumption at the RIS is considered i.e., same number of activated elements, between the AR scheme, with $Q=16$ elements and $K=4$ partitions, and the PB case with $Q=4$. In this case, we observe that the performance of the two scenarios is comparable. This shows that the proposed schemes have low implementation complexity and conserve the available resources at the RIS, but can still achieve significant performance gains.

	Finally, a comparison of the achieved DMT between the PR scheme and the partition-based RIS schemes is depicted in Fig. \ref{fig:dmt}, for a network setting with $N=L=3$ antennas. We can see that for the PR scheme ($K=1$) the optimum DMT is achieved when $Q=N+L-1=5$, which is denoted by the red solid line. Any other topology with $Q>5$ can be vertically reduced to the $(3,5,3)$ channel, since it will be DMT-equivalent. By employing the AR scheme, we observe that the DMT can be significantly improved in terms of diversity gain. Moreover, the proposed scheme achieves the maximum multiplexing gain, if the number of partitions satisfies the lower bound of \eqref{cond_dmt} i.e., $m\geq\min \{N,L\}$. However, under the AR scheme, the considered channel setting can not achieve both $d_{\max}$ and $r_{\max}$ provided by the cut-set bound with a single value of $K$, since \eqref{cond_dmt} can not be fully satisfied with any partition size. Therefore, if a larger number of partitions is considered so that $m<\min \{N,L\}$, the presented scheme can still increase the achieved diversity gain, but it becomes suboptimal in terms of multiplexing gain. This remark is illustrated in Fig. \ref{fig:dmt}, by considering an RIS with $Q=10$ and a partition size of $5$ or $10$ sub-surfaces. Fig. \ref{fig:dmt} also shows the provided lower bound on the DMT achieved by the FR scheme. It can be seen that the FR scheme outperforms both the PR and the AR scheme. In particular, with the FR scheme, the RIS-aided channel achieves the same diversity order obtained by the AR scheme, and remains optimal in terms of multiplexing gain, for all partition sizes. Finally, for the considered network topology, only the FR scheme achieves both $d_{\max}$ and $r_{\max}$ with a single value of $K$, which is depicted with the black dotted line for the case with $Q=K=10$.

\section{Conclusions}\label{conc}
	In this paper, we studied the performance of RIS-assisted MIMO communications by employing two low-complexity partition-based schemes, that do not require any CSI knowledge at the transmitter side for their implementation. Specifically, the RIS elements are partitioned into sub-surfaces, which are reconfigured periodically in an efficient manner, creating a parallel channel in the time domain. We first proposed the AR scheme, which activates each sub-surface in a consecutive order. Next, we presented the FR scheme, where each sub-surface sequentially flips the incident signals through a phase shift adjustment. Theoretical expressions of the outage probability and the DMT were provided, while the considered schemes were compared to the conventional PR scheme. We demonstrated that, by considering the proposed schemes, we can increase the diversity order achieved by the RIS-aided system and improve the performance beyond the limit of the MIMO channel, while at the same time we keep the implementation complexity low.

	\begin{figure}[t!]
		\centering
		\includegraphics[width=0.95\linewidth]{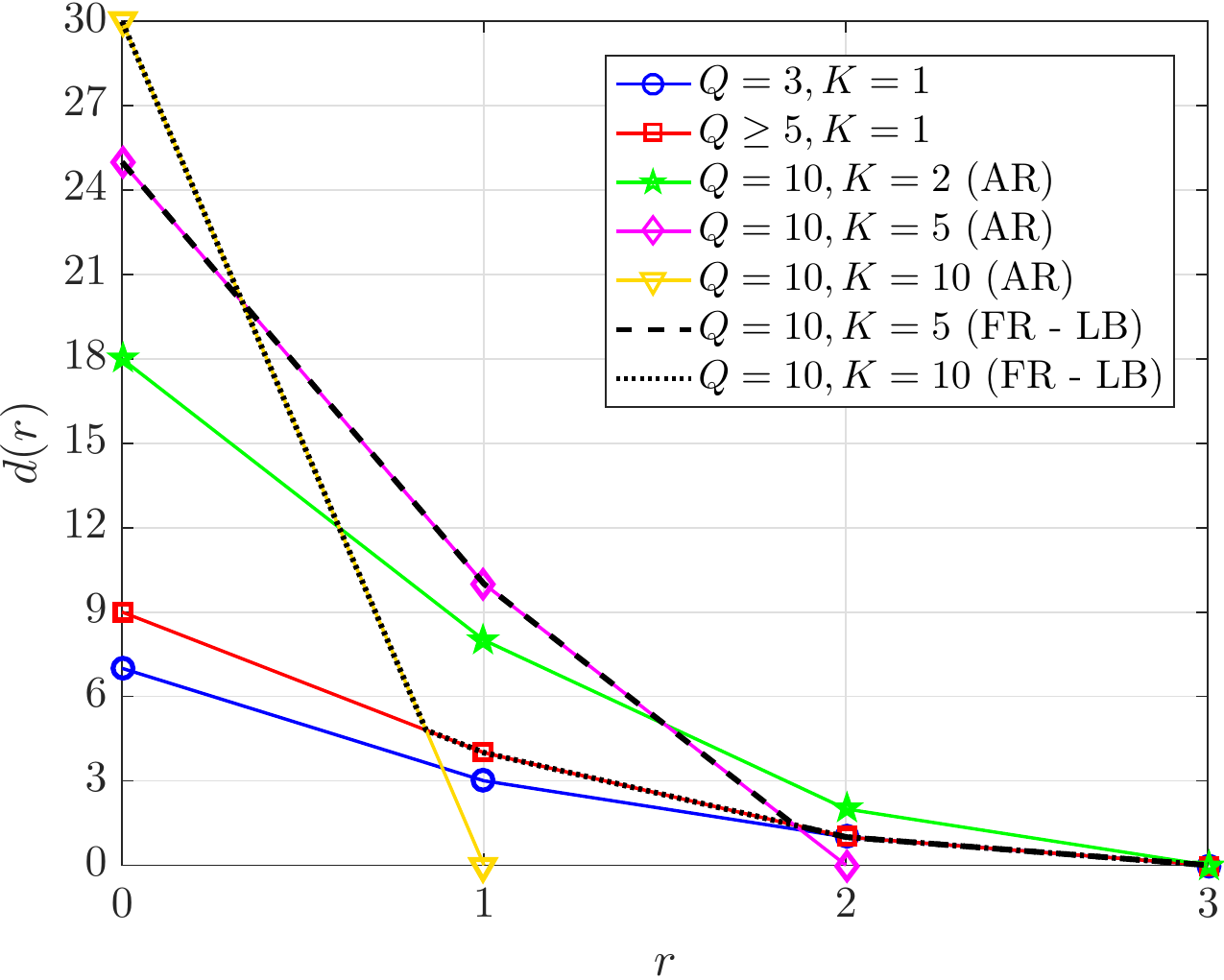}
		\caption{DMT comparison of the PR scheme and the partition-based schemes for a network topology with $N=L=3$.}
		\label{fig:dmt}
	\end{figure}

\appendix
	\subsection{Proof of Lemma \ref{char_fun}} \label{proofA}
		Since $K=1$, all the RIS elements belong to the same (single) sub-surface and are reconfigured only at the beginning of each time slot. Due to the random rotations at the elements, the induced phase shifts do not have any effect on the channel gain \cite{psomas2021}. Therefore, without loss of generality, we can equivalently consider the case where $\phi_{i,1}=0$, with $1\leq i\leq Q$. The reflection matrix of the RIS is then equal to the $Q\times Q$ identity matrix, $\Phi=\mathbf{I}_Q$, and the resulting ($N,Q,L$) channel can be interpreted as a Rayleigh product channel i.e., $\mathcal{H}_1=\mathbf{G}\mathbf{H}$. The mutual information of the channel is thus equal to
		\begin{equation}
			\mathcal{I}=\log_2\left[ \det\left( \mathbf{I}_{L}+\dfrac{\rho}{N}\mathcal{H}_1\mathcal{H}_1^{\dagger}\right) \right]=\sum_{i=1}^{n_0}\log_2\left( 1+\dfrac{\rho}{N}\lambda_i\right),
		\end{equation}
		where $\lambda_i$, $1\leq i\leq n_0$, are the eigenvalues of $\mathcal{H}_1\mathcal{H}_1^{\dagger}$ and their joint PDF is given by \cite{taricco2017}. The characteristic function of the mutual information is defined as
		\begin{equation}
			\varphi(Q,t)=\mathbb{E}\left\lbrace \exp(\jmath t\mathcal{I})\right\rbrace=\mathbb{E}\left\lbrace \prod_{i=1}^{n_0}\left(1+\dfrac{\rho}{N}\lambda_i\right)^{\jmath t/\ln 2}\right\rbrace.
		\end{equation}
		By using the analytical results of \cite[Section II.B]{taricco2017} for the evaluation of the moment generating function of the mutual information of a multi-hop MIMO channel, the final expression of $\varphi(Q,t)$ is obtained.

	\subsection{Proof of Theorem \ref{pout_AR_CLT}} \label{proofB}
		By considering the SISO channel, the end-to-end channel matrix given in \eqref{eff_channel} is simplified to a scalar value which, under the AR scheme, is given by
		\begin{equation}
			\mathcal{H}_k=\sum_{i=1}^{Q}h_{i}g_{i}a_{i,k},
		\end{equation}
		where $a_{i,k}$ is provided according to the framework of the AR scheme in \eqref{amplitude_AR}. It has been proved in \cite{psomas2021} that by applying the CLT, the distribution of the end-to-end channel $\mathcal{H}_k$ converges to a complex Gaussian distribution with zero mean and variance $m$, since at sub-slot $k$ only the $m$ elements that correspond to the sub-surface $\mathcal{S}_k$ are turned ON. Therefore, the channel gain $W_{k}=\left| \mathcal{H}_k\right| ^2$ is exponentially distributed with parameter $1/m$. Recall that the sub-surfaces $\mathcal{S}_k$ do not overlap, so the random variables $W_{k}$ are mutually independent. The outage probability is then calculated as follows
		\begin{align}
			\Pi_{\AR}^{\CLT}(R,K,m)&=\mathbb{P}\left\lbrace \dfrac{1}{K}\sum_{k=1}^{K}\log_2\left( 1+\rho W_k\right) < R\right\rbrace\nonumber\\
			&=\mathbb{P}\left\lbrace\prod_{k=1}^{K}\left( 1+\rho W_k\right) < 2^{RK}\right\rbrace,
		\end{align}
		which follows from the logarithmic identity $\log_2(x)+\log_2(y)=\log_2(xy)$. The final expression of the outage probability in \eqref{pout_approx_AR} is eventually derived by obtaining the cumulative distribution function (CDF) of the product of $K$ independent shifted exponential random variables, which is given in \cite[Corollary~2]{yilmaz2009}.

	\subsection{Proof of Lemma \ref{corr_coeff}} \label{proofC}
		Since $N=L=1$, the channel $\mathcal{H}_k$ is a scalar value given by
		\begin{equation}
			\mathcal{H}_k=\sum_{i=1}^{Q}h_{i}g_{i}\exp{\jmath\phi_{i,k}},
		\end{equation}
		where $\phi_{i,k}$ is provided in \eqref{phase_FR} by employing the FR scheme. The random variables $W_{k}=\left| \mathcal{H}_k\right| ^2$ are correlated since the channel coefficients $h_i$ and $g_i$ remain constant during one time slot i.e., over $K$ sub-slots. In order to calculate the correlation coefficient between $W_k$ and $W_l$, $k\neq l$, we consider the Pearson correlation formula which is given by
		\begin{equation}\label{pearson}
			\zeta=\dfrac{\mathbb{E}\left\lbrace W_{k}W_{l}\right\rbrace-\mathbb{E}\left\lbrace W_{k}\right\rbrace\mathbb{E}\left\lbrace W_{l}\right\rbrace}{\sigma_{W_{k}}\sigma_{W_{l}}},
		\end{equation}
		where $\sigma_{W_{i}}=\sqrt{\mathbb{E}\left\lbrace W_{i}^{2}\right\rbrace-\mathbb{E}\left\lbrace W_{i}\right\rbrace^{2}}$. We find that $\mathbb{E}\left\lbrace W_{i}\right\rbrace=Q$, $\mathbb{E}\left\lbrace W_{i}^{2}\right\rbrace=2Q(Q+1)$ and $\mathbb{E}\left\lbrace W_{k}W_{l}\right\rbrace=Q(Q+3)+(Q-b)(Q-3b-1)+b(b-1)$, where $b$ is given by \eqref{b_value}. By substituting the previous results in \eqref{pearson}, and after some trivial algebraic manipulations, we get the final expression of $\zeta$ as in \eqref{corr_expr}.

	\subsection{Proof of Theorem \ref{approx_FR_thm}} \label{proofD}
		By using the approximated channel $\mathcal{\tilde{H}}_{k}$ defined in \eqref{corr_Rayleigh} and by setting $W_k=| \mathcal{\tilde{H}}_{k}| ^{2}$, the outage probability under the FR scheme is evaluated as
		\begin{align}
		&\Pi_{\FR}(R,K,m)\nonumber\\
		&=\mathbb{P}\left\lbrace \dfrac{1}{K}\sum_{k=1}^{K}\log_2\left( 1+\rho W_k\right) < R\right\rbrace\nonumber\\
		&=\mathbb{P}\left\lbrace\prod_{k=1}^{K}\left( 1+\rho W_k\right) < 2^{RK}\right\rbrace\label{ineq}\\
		&=\mathbb{E}_{W_k}\left\lbrace F_{W_K}\left[ \dfrac{1}{\rho}\left(\dfrac{2^{RK}}{\prod_{k=1}^{K-1}\left( 1+\rho W_k\right)}-1\right)  \right] 	\right\rbrace ,
		\end{align}
		which follows by solving the inequality in \eqref{ineq} for $W_K$. The conditional CDF of $W_K$, given $W_1$, is derived by taking the integral of \eqref{cond_pdf_W_k}, which is equal to
		\begin{equation}
			F_{W_K|W_1}(x|y)=1-\mathcal{Q}_{1}\left( \sqrt{\dfrac{2\zeta y}{\Omega}},\sqrt{\dfrac{2x}{\Omega}}\right).
		\end{equation}
		Since the random variables $W_k$, conditioned on $W_1$, are all independent between each others, we have
		\begin{multline}
			\Pi_{\FR}(R,K,m)=\int_{w_1}\cdots\int_{w_{K-1}}f_{W_1}(w_1)\\
			\times\prod_{k=2}^{K-1}f_{W_k|W_1}(w_k|w_1)F_{W_K|W_1}\left(\vartheta\middle|w_1\right)dw_{K-1}\cdots dw_1,
		\end{multline}
		where
		\begin{equation}
			\vartheta\triangleq \dfrac{1}{\rho}\left(\dfrac{2^{RK}}{\prod_{k=1}^{K-1}\left( 1+\rho w_k\right)}-1\right),
		\end{equation}
		while $f_{W_1}(w_1)$ and $f_{W_k|W_1}(w_k|w_1)$ are given by \eqref{pdf_W_1} and \eqref{cond_pdf_W_k}, respectively. Regarding the integration limits, we need to ensure that the inequality
		\begin{equation}
			\dfrac{2^{RK}}{\prod_{k=1}^{K-1}\left( 1+\rho w_k\right)}-1>0,
		\end{equation}
		is satisfied sequentially for each $w_k$ from $K-1$ to $1$. The final expression is derived by using the integral transformation $\tau_k\rightarrow 1+\rho w_k$.

	\subsection{Proof of Proposition \ref{prop_dmt_FR}} \label{proofE}
		We first prove that the FR scheme achieves the same diversity order as the AR scheme i.e., 
		\begin{equation}\label{div_order_AR_FR}
			d_{(N,Q,L)}^{\FR}(0)=d_{(N,Q,L)}^{\AR}(0)=Kd_{(N,m,L)}^{\PR}(0).
		\end{equation}
		Let us denote by $\mathcal{H}_k\sp{\prime}$ the $k$-th sub-channel created under the FR scheme, and by $\mathcal{H}_k$ the respective sub-channel resulting from the AR scheme. Based on the description of each scheme, the set of matrices $\left\lbrace \mathcal{H}_k\sp{\prime}\right\rbrace _{k=1}^K$ is derived from $\left\lbrace \mathcal{H}_k\right\rbrace _{k=1}^K$ by using the transformation
		\begin{equation}
			\left[ \mathcal{H}_1\sp{\prime} \: \mathcal{H}_2\sp{\prime} \: \ldots \mathcal{H}_K\sp{\prime}\right] = \left[ \mathcal{H}_1 \: \mathcal{H}_2 \: \ldots \mathcal{H}_K\right] \mathbf{T},
		\end{equation}
		where $\mathbf{T}$ is a $KN\times KN$ matrix composed of $K\times K$ sub-matrices and each sub-matrix $\mathbf{T}_{i,j}$ is equal to
		\begin{equation}
			\mathbf{T}_{i,j}=
			\begin{cases}
				-\mathbf{I}_{N}, & (i=j,K>2)\: \text{or} \: (i=j>1,K=2);\\
				\mathbf{I}_{N}, & \text{otherwise}.
			\end{cases}
		\end{equation}
		We can easily verify that $\mathbf{T}$ is invertible, constant and linear. Therefore, the two schemes achieve the same diversity order since
		\begin{align}
			\sum_{k=1}^{K}\left\| \mathcal{H}_k\sp{\prime}\right\| _{F}^{2}
			&\geq \lambda_{\min}\left( \mathbf{T}\mathbf{T}^{\dagger}\right)\sum_{k=1}^{K}\left\| \mathcal{H}_k\right\| _{F}^{2}= K\sum_{k=1}^{K}\left\| \mathcal{H}_k\right\| _{F}^{2}\nonumber\\
			&\doteq\sum_{k=1}^{K}\left\| \mathcal{H}_k\right\| _{F}^{2},
		\end{align}
		where $\lambda_{\min}(\cdot)$ denotes the minimum eigenvalue of a matrix and the relation $a\doteq b^c$ means $\lim_{b\rightarrow \infty} \frac{\log a}{\log b}=c$.

		The next step is to prove that the FR scheme can achieve at least the DMT of the AR scheme for $r=0,...,\min\{N,m,L\}$. In the considered channel, the multiplexing gain $r$ can be achieved if an ($r,r,r$) sub-channel is reserved for spatial multiplexing at each time sub-slot. According to Corollary \ref{dmt_AR} and by using the DMT characterization of a Rayleigh product channel given in \cite[Theorem~3.4]{yang2007}, we have
		\begin{equation}
			d_{(N,Q,L)}^{\AR}(r)=Kd_{(N-r,m-r,L-r)}^{\PR}(0)=d_{(N-r,Q-Kr,L-r)}^{\AR}(0).
		\end{equation}
		Similar to \eqref{div_order_AR_FR}, the ($N-r,Q-Kr,L-r$) channel achieves the same diversity order under both the FR and the AR scheme. Besides, one can show that $d_{(N,Q,L)}^{\FR}(r)\geq d_{(N-r,Q-Kr,L-r)}^{\FR}(0)$, since all the elements are activated in the FR scheme and the multiplexing gain $r$ can be achieved with less than $Kr$ elements. It follows that
		\begin{equation}\label{LB_AR_FR}
			d_{(N,Q,L)}^{\FR}(r)\geq d_{(N,Q,L)}^{\AR}(r).
		\end{equation}
		Finally, the parallel channel under the FR scheme is in outage if $\sum_{k=1}^{K}\mathcal{I}_k<Kr\log\rho$, which implies that at least one of the sub-channels is in outage with the corresponding random variable $\mathcal{I}_k$ below the target rate $r\log\rho$. Therefore, the FR scheme can achieve at least the DMT of the conventional case i.e.,
		\begin{equation}\label{LB_PR_FR}
			d_{(N,Q,L)}^{\FR}(r)\geq d_{(N,Q,L)}^{\PR}(r).
		\end{equation}
		By combining \eqref{div_order_AR_FR}, \eqref{LB_AR_FR} and \eqref{LB_PR_FR} we get the lower bound for the DMT of the FR scheme as \eqref{dmt_FR_LB}.

\end{document}